\documentclass[10pt]{article}

\usepackage[ruled,vlined,norelsize]{algorithm2e}

\usepackage{authblk}
\usepackage[utf8]{inputenc}    
\usepackage[T1]{fontenc}

\usepackage{graphicx}
\usepackage{amsthm,amsmath}
\usepackage{amssymb}
\usepackage{amsfonts}
\newtheorem{theorem}{Theorem}

\newtheorem{lemma}[theorem]{Lemma}

\newtheorem{conjecture}{Conjecture}

\usepackage{tikz}
\usepackage{comment}
\usepackage[absolute,overlay]{textpos}         
\DeclareGraphicsExtensions{.png,.pdf,.eps}

\providecommand{\keywords}[1]{\textbf{\textit{Keywords:}}#1}

\newcommand\blfootnote[1]{%
  \begingroup
  \renewcommand\thefootnote{}\footnote{#1}%
  \addtocounter{footnote}{-1}%
  \endgroup
}

\title{Mixing Time for Square Tilings \footnote{This work was supported by the ANR project QuasiCool (ANR-12-JS02-011-01)}}
\author{Alexandra Ugolnikova\footnote{contact: alexandra.ugolnikova \textit{at} lipn.univ-paris13.fr}}

\affil{Laboratoire d'Informatique de Paris Nord}
\date{\today}

\begin{document}

\maketitle

\begin{abstract}
We consider tilings of $\mathbb{Z}^2$ by two types of squares. We are interested in the rate of convergence to the stationarity of a natural Markov chain defined for square tilings. The rate of convergence can be represented by the mixing time which measures the amount of time it takes the chain to be close to its stationary distribution. We prove polynomial mixing time for $n \times \log n $ regions in the case of tilings by $1 \times 1$ and $s \times s$ squares. We also consider a weighted Markov chain with weights $\lambda$ being put on big squares. We show rapid mixing of $O(n^4 \log n)$ with conditions on $\lambda$. We provide simulations that suggest different conjectures, one of which is the existence of frozen regions in random tilings by squares.
\end{abstract}
\blfootnote{\textup{2010} \textit{Mathematics Subject Classification}: 60J10, 52C20.}

\section{Introduction}

In the present work we consider tilings of a closed simply connected region of $\mathbb{Z}^2$ by two squares of different sizes $k$ and $s$, $k,s \in \mathbb{N}$, that we denote by $(k,s)-$tilings (Figure \ref{fig: squares}) shows an example of different square tilings of a $60\times 60$ square region). We are interested in studying their structure and answering questions about random generation.

In Section \ref{sec: preliminaries} we lay out the necessary notions and results about Markov chains with which we operate throughout other sections. In Section \ref{sec : settings} we define a height function for the square tilings and local transformations (flips) following definitions and results by \cite{CL90, KK92, K94,BR02}, followed by a series of examples. At the end of Section \ref{sec : settings} we introduce a Markov chain MC$_{square}$ defined for our tiling system. Ideally, we want to be able to rigorously bound the time that it takes for the chain to reach stationarity (its mixing time). 

In Section \ref{subsec: fast mixing} we prove rapid mixing for MC$_{square}$ (polynomial over the size of the tileable region) for $(1,s)$ tilings of $n \times \log n $ regions (Theorem \ref{thm: mixing time for nxlong }). We conjecture that it is fast mixing for any $m$ and $s$ when considered for these long regions. 

In Section \ref{subsec: weighted} we consider a weighted version of the chain, where weights $\lambda$ are assigned to one type of squares in a tiling. Adding weights helps to notice phase transitions in the system. Usually the weights $\lambda$ are assigned in such a way that the case $\lambda = 1$ corresponds to the unweighted version of the chain. This case is generally hard to analyze but it becomes easier to analyze the dynamics for $\lambda$ below and above some critical point. We draw polynomial bounds on the mixing time of the Markov chain for $(1,2)-$tilings with a condition on $\lambda$ in Theorem \ref{th: weighted mixing} and for $(1,s)-$tilings with a condition on $\lambda$ in Theorem \ref{thm: $(1,s)$ weighted coupling}.

In Section \ref{subsec: simulations} we present simulations and conjectures that show that the mixing time of the unweighted version might not be polynomial but sub-exponential, and might be as difficult to analyze as critical cases of Markov chains for independent sets and perfect matchings. A relation to the independent sets is discussed in Section \ref{subsec: weighted}. 

At the end, in Section \ref{subsec: limit shape}, we conjecture an analog of the Aztec diamond for dominoes (hexagon for lozenges) and present simulations that show an ``Arctic circle''--type phenomenon.

\section{Preliminaries: Mixing and coupling times}
\label{sec: preliminaries}


We consider reversible ergodic Markov chains with a finite state space $\Omega$. We denote its stationary distribution by $\pi$, its probability law by $P$. For any initial state $x \in \Omega$ let the \textbf{total variation distance} between $P(x, \cdot)$ and $\pi$ is 

$$
d_{TV}(P(x, \cdot), \pi) := \frac{1}{2} \sum_{y \in \Omega} \lvert P^t(x,y) - \pi(y) \rvert.
$$

Let us write it as $d_x(t)$. The \textbf{mixing time} of a MC is the time it takes the chain to get close to its stationary distribution. Formally, it is defined as follows:

$$
\tau_{mix}(\varepsilon) := \max_{x \in \Omega} \min \lbrace t: d_x(t') \leq \varepsilon \;\forall\; t' \geq t\rbrace , 
$$

$$
\tau_{mix}:= \tau_{mix} \left( \frac{1}{4} \right).
$$

A classical way to bound the rate of convergence of a chain is to bound its mixing time. There are lot of different ways of bounding the mixing time: via the second largest eigenvalue (which can be analyzed using the corresponding tiling graph's properties, although it often turns out to be difficult due to the unknown graph's structure), coupling methods (see, e.g., \cite{BD97, DS91, DG98, LPW09, R06,S92}). 

Here we concentrate on the \textbf{coupling} method.
A coupling for two probability distributions $\mu$ and $\nu$ is a pair of random variables $(X,Y)$ defined on the  same probability space such that $\mathbb{P}(X=x) = \mu(x)$ and $\mathbb{P}(Y=y) = \nu(y)$. Here we  will be  using couplings for Markov chains  where constructing copies of the chain proves to be a useful tool to analyze the distance to stationarity. A \textbf{coupling of a MC} is a stochastic process $(X_t, Y_t)_t$ on $\Omega \times \Omega$ such that:
\begin{enumerate}
\item $X_t$ and $Y_t$ are copies of the MC with initial states $X_0 = x$ and $Y_0 = y$;
\item If $X_t = Y_t$, then $X_{t+1}  = Y_{t+1}$.
\end{enumerate}
Let $T^{x,y} = \min \lbrace t: X_t = Y_t \vert X_0 = x, Y_0 = y\rbrace$. Then define the \textbf{coupling time} of the MC to be
$$
\tau_{cp} :=  \max_{x,y} \mathbb{E} T^{x,y}.
$$

The following result \cite{A81} relates the coupling and mixing times:

\begin{theorem}[Aldous]\label{theorem: A}
$$
\tau_{mix}(\varepsilon) \leq \lceil \tau_{cp} e\ln \varepsilon^{-1}  \rceil.
$$
\end{theorem}

One of the most used methods to bound the mixing time is the following path coupling theorem \cite{DG98}. The authors show that in order to bound the coupling time, one only has to consider pairs of configurations of the coupled chain that are close to each other in the defined metric. It is sufficient to prove that they [each pair of configurations] have more tendency to remain close to each other under the evaluation of the chain. Then the mixing time is polynomial and depends on the diameter of the corresponding graph. 

\begin{theorem}[Dyer-Greenhill]\label{theorem: DG}
Let $\varphi : \Omega \times\Omega \rightarrow \lbrace 0, \ldots, D \rbrace$  be an integer-valued metric,  $U$ -- a subset of $\Omega \times \Omega$ such that for all $(X_t, Y_t) \in  \Omega \times \Omega $ there exists a path between them: $X_t  = Z_0, Z_1, \ldots, Z_n = Y_t $  with $(Z_i, Z_{i+1}) \in U$ for $ 0\leq i\leq r-1 $ and $$ \sum_{i=0}^{r-1} \varphi(Z_i, Z_{i+1}) = \varphi(X_t, Y_t).$$
Let MC be a Markov chain on $\Omega$ with transition matrix $P$. Consider a random function $f: \Omega \rightarrow \Omega$ such  that $\mathbb{P}(f(X)= Y) = P(X,Y) $ for all $X,Y \in \Omega$, and let a coupling be  defined by $(X_t,Y_t) \rightarrow (X_{t+1}, Y_{t+1}) = (f(X_t), f(Y_t))$.

\begin{enumerate}
\item If there exists $\beta < 1$  such that $ \mathbb{E}[\varphi(X_{t+1}, Y_{t+1})] \leq \beta \varphi(X_t,Y_t)$ for all $(X_t, Y_t) \in U$, then the mixing time satisfies $$ \tau_{mix}(\varepsilon) \leq \frac{\ln(D\varepsilon^{-1})}{1-\beta}.$$ 
\item If $\beta = 1$ and there exists $\alpha > 0$ : $\mathbb{P}(\varphi(X_{t+1},Y_{t+1}) \neq  \varphi(X_t,  Y_t)) \geq \alpha $ for all $t$ such that $X_t \neq Y_t$. Then the mixing time  satisfies  
$$ \tau_{mix}(\varepsilon) \leq  \left\lceil\frac{eD^2}{\alpha} \right\rceil \lceil \ln \varepsilon^{-1} \rceil .$$ 
\end{enumerate}
\end{theorem} 

\section{Settings}
\label{sec : settings}

\subsection{Height function}
 \label{subsec: height f}

Height functions for tilings by rectangular tiles appear in numerous works (e.g., \cite{CL90, KK92, K94, BR02}) that use Conway tiling groups \cite{T90}. Having a well defined height function permits to get a structure on the set of configurations and then use it to construct an algorithm that verifies the tileability of a given region.

Let $a$ be a symbol corresponding to a horizontal step of length one to the right on the rectangular grid, $a^{-1} - $ to a step to the  left, $b -$ to a vertical step up and $b^{-1} -$ to a step down. Then any path on the skeleton of any tiling of $R$ is a word in the alphabet $\lbrace a, a^{-1},b, b^{-1} \rbrace$, as well as the perimeter of $R$.

Consider a group $G$ generated by $a$ and $b$ and restrictions that ensure that a path around each tile is elementary. A path around an $m \times m $ square is  $ a^mb^ma^{-m}b^{-m} =:[a^m, b^m]$. Similar for the other tile. Then the tiling group $G$ is defined as follows:  $G = \left\langle a, b \mid [a^m, b^m], [a^s, b^s] \right\rangle $, where  $[a^m, b^m] =[a^s, b^s] = e$. Any path $g$ on a skeleton of a tiling of $R$ can be expressed using elements of $G$. For example, consider a tiling  of a $2m \times 2m$ rectangle by  $m \times m $ squares. Let $l$ be the left lower corner of the region, $r$  be its rights upper corner. Then
$ g(l,r) = a^mb^ma^mb^m$ is a path (one of many) from $l$ to $r$. Another path is, e.g., $g'(l,r) = a^{m}a^{m}b^{m}b^{m} = a^{2m}b^{2m}$.

\subsubsection{Quotient group H}

\begin{figure}
\centering
\includegraphics[scale=0.8]{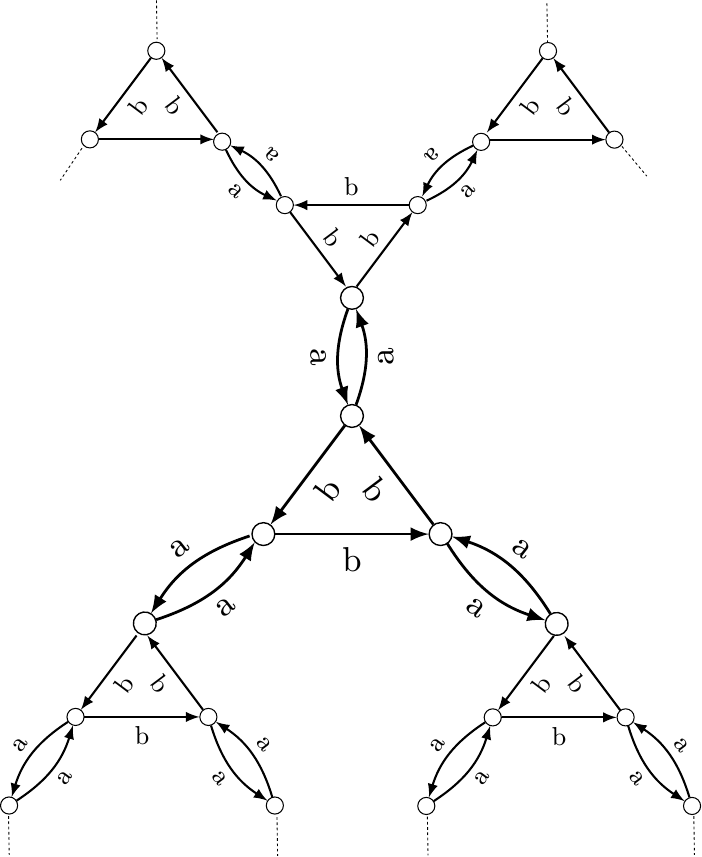}
\caption{Cayley graph of $H \simeq \mathbb{Z}_2 \times \mathbb{Z}_3$.}
\label{fig: Cayley 2x3}
\end{figure}


In order to define a height function, we introduce a quotient group $H$ of $G$. Let $H = G/ \left\langle a^{m}, b^{s} \right\rangle$. Since  $a^{m}= b^{s}= e$, then $a^{m}b^{m} =b^{m}a^{m}$ and $a^{s}b^{s} =b^{s}a^{s}$. So $H = \left\langle a, b \mid a^m, b^s \right\rangle \simeq \mathbb{Z}_m \times \mathbb{Z}_s$. Let $\Gamma_H$ be the Cayley graph of $H$. It is a tree made of cycles of size $m$ with cycles of size $s$ attached to every vertex.
Example for $m=2$, $s=3$ is show in Figure \ref{fig: Cayley 2x3}.

Every path $g$ has a unique (canonical) expression using elements of $\mathbb{Z}_m \ast \mathbb{Z}_s$:
$$g = k_1l_1k_2l_2 \ldots k_rl_r, $$
where $k_i \in \mathbb{Z}_m, l_i \in \mathbb{Z}_s$ for all $1 \leq i \leq r$ and some non-negative $r$. It can easily be constructed following the edges of the Cayley graph.

\subsubsection{Weighted Cayley graph}

\begin{figure}
\centering
\includegraphics[scale=0.8]{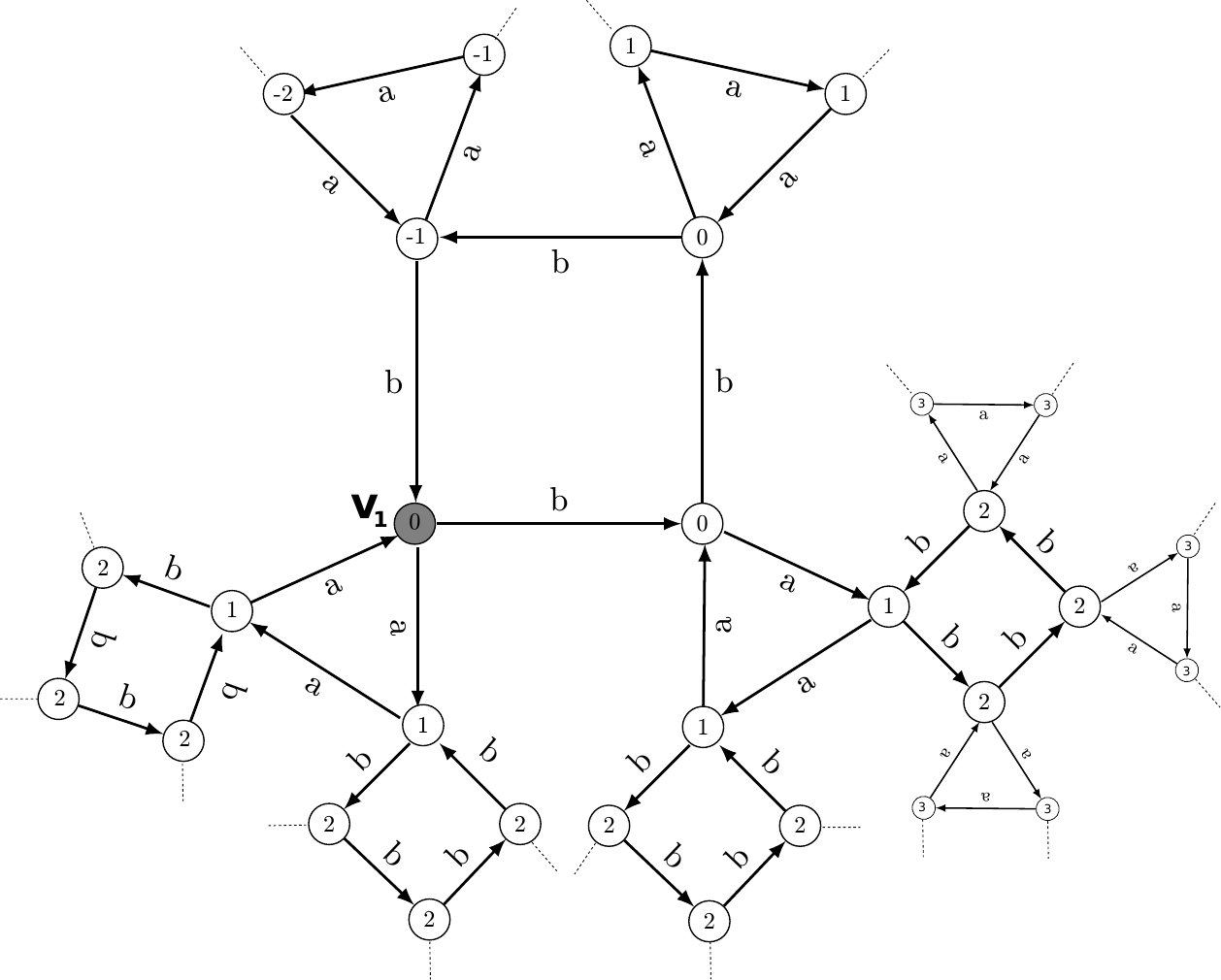}
\caption{Weighted Cayley graph for $H \simeq \mathbb{Z}_3 \times \mathbb{Z}_4$.}
\label{fig: Cayley weighted 3x4}
\end{figure}

Consider the Cayley graph $\Gamma_H$. Choose a vertex $v_1$, let it be the starting point and set its \textbf{weight} $w(v_1) := 0$. $v_1$ belongs to two cycles $C^m_{v_1}$ of size $m$ and $C^s_{v_1}$ of size $s$. $C^s_{v_1} = \lbrace v_1, v_2, \ldots v_s \rbrace$. Set
\begin{itemize}
\item $w(v_i) = w(v_1)$ for $i = 2, \ldots, s-1$,
\item $w(v_s) = w(v_1)-1$.
\end{itemize}
 Let $C^m_{v_1}, C^m_{v_2}, \ldots, C^m_{v_s} $ be cycles of size $m$ to which vertices $v_1, \ldots v_s$ belong accordingly. For all $i = 1, \ldots, s-1$, set
 \begin{itemize}
 \item $w(y) = w(v_i)+1$ for all $y \in C^m_{v_i}$, $y \neq v_i.$
 \end{itemize} 
 For $i=s$, set
 \begin{itemize}
 \item  $w(z_j) = w(v_s)$ for  $z_j \in C^m_{v_s} \setminus \lbrace v_s \rbrace $ where $j = 1, \ldots m-2$,
 \item $w(z_{m-1}) = w(v_s)-1$. 
  \end{itemize}

Continue assigning weights to the vertices of the Cayley graph in the way described above. The graph is infinite, a small part of the Cayley graph for $m=3, s=4$ is shown in Figure \ref{fig: Cayley weighted 3x4}.

Another way to define the above is the following:
\begin{itemize}
\item In each cycle all vertices are of weight $w$, except for one that has weight $w-1$, let us call it \textbf{descending}.
\item Two cycles are connected via a vertex that is descending in one cycle and is not descending in the other.
\end{itemize}


\subsubsection{Height function for $(m,s)-$square tilings}

\begin{figure}
\centering
\includegraphics[scale=1]{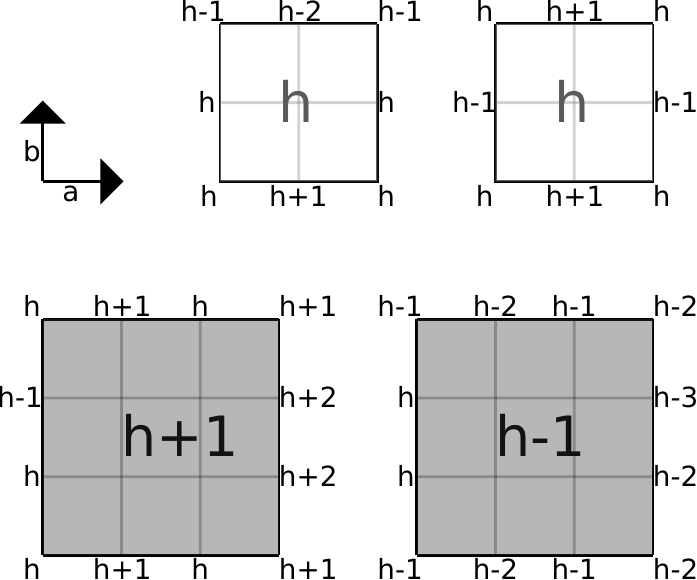}
\caption{Heights for $t_2$ and $t_3$ tiles in $(2,3)-$tilings.}
\label{fig: square_2_3_tile_height}
\end{figure}

For each $(m,s)-$ tiling $T$ of a s.c. finite region $R$ of $\mathbb{Z}^2$ denote the set of vertices of its skeleton by $V_T$. The height $h$ is defined on $ V_T$ as follows. Choose an initial point $x$ on the boundary of $R$ (if $R$ is a rectangle $x$ can be chosen as the left lower corner of $R$ for example). Without loss of generality, set $h(x) := w(v_1) = 0$.  To calculate the height of any given $y$, construct a path $g(x,y)$ connecting $x$ and $y$ in the following way: start the path in $x$ and make moves that do not cross any tile. Every time a horizontal or vertical move is made, do the corresponding move in the weighted Cayley graph. The weight of the final node gives $h(y)$ and does not depend on the choice of the path. Let us point out that the height function on the boundary is completely defined by the region itself and does not depend on the tiling.

For $m,s\geq 2$ the height function $h$ defined on vertices of $(m,s)-$ square tilings of a rectangular region $R$ is flat on the boundary (it goes around the same two cycles on the Cayley graph) and $O($size of $R)$ on the interior points. For the degenerate case when $m$ or $s$ equals one, $h$ is simply a constant.

There are two tiles: $t_s $ is the square tile of size $s$, $t_m -$ of size $m$. Consider $t_m$. Since elements of $\mathbb{Z}_m$ encode the horizontal steps, the heights of vertices on one of the vertical sides are exactly the  same as the heights on the other vertical side. Let $h$ be the maximum over heights of vertices on the vertical side, then the maximum over heights of vertices on each of the horizontal sides belongs to $\lbrace h-1, h+1 \rbrace$. Define the height of the tile $t_m$ $h(t_m):=h$. Change vertical sides  to horizontal sides to get the similar  property for $t_s$. Figure \ref{fig: square_2_3_tile_height} shows how the heights are defined each node for $t_2$ and $t_3$ using the weighted Cayley graph where weights are used to write heights in each vertex.

\begin{figure}
\centering
\includegraphics[scale=1]{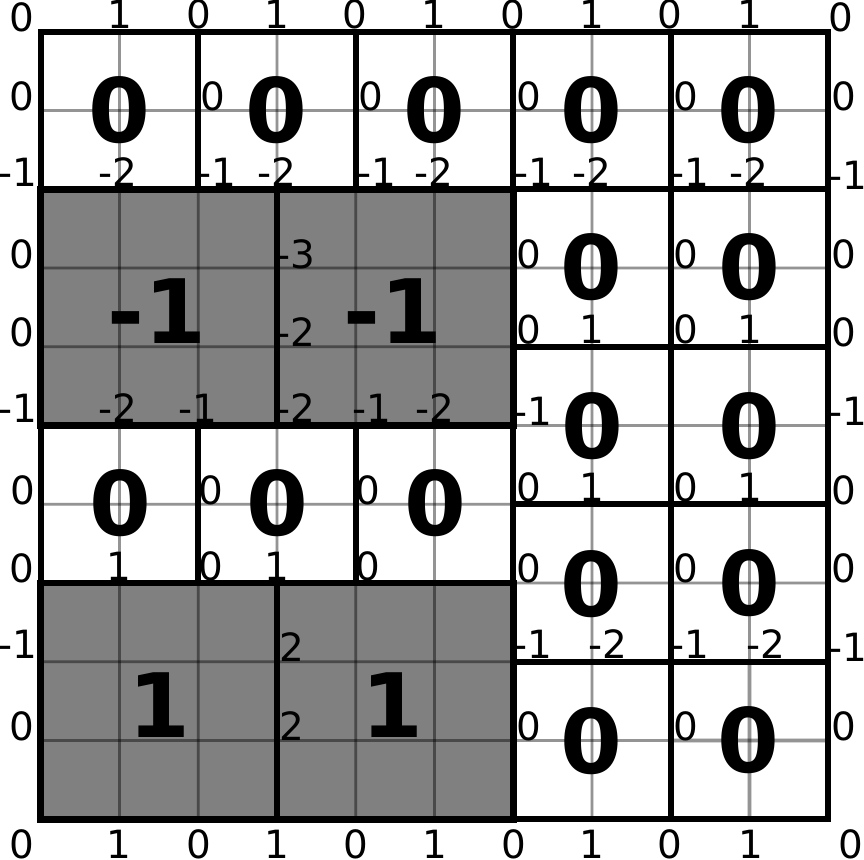}
\caption{A $(2,3)-$tiling of a $10 \times 10$ square with heights.}
\label{fig: 23_tiling_10}
\end{figure}

Let $h(T)$ be the height of a tiling $T$. Define it as follows:
$$
h(T) = \sum_{t \in T} a(t)h(t),
$$
where $a(t)$ is the area of a tile $t$.  An example of a $(2,3)-$tiling of a $10 \times 10$ square with heights is shown in Figure\ref{fig: 23_tiling_10}. Its total height is $ 1 \times 9 + 1 \times 9 + (-1)\times 9 + (-1)\times 9  = 0$.

For the degenerate case $h(T)$ is not very interesting because $h(T) \equiv const$ area$(R)$ for every tiling $T$.

\subsection{Flips}
\label{subsec: flips}

\begin{figure}
\centering
\includegraphics[scale=0.5]{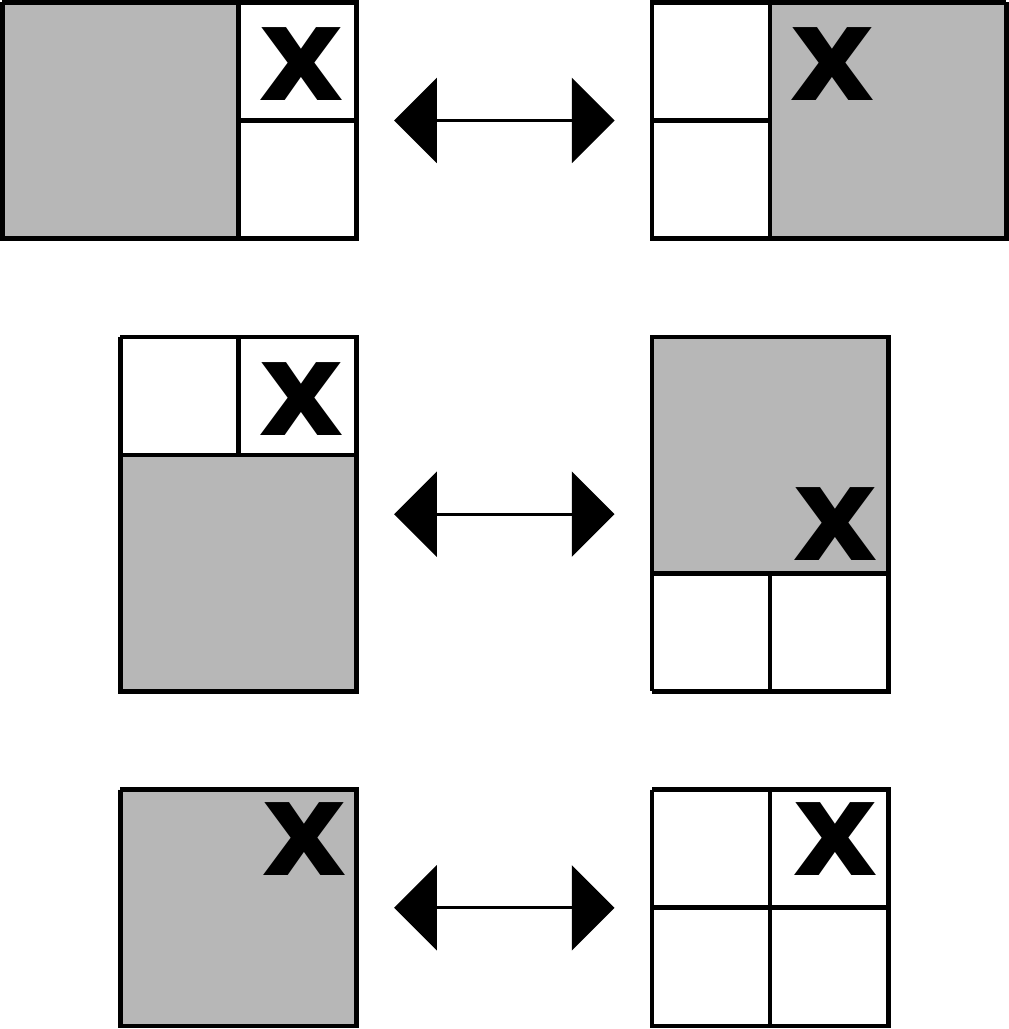}
\caption{Flips for the $(1,2)-$ tiling.}
\label{fig: flip12}
\end{figure}

Let us define local transformations, or \textbf{flips}, for $(m,s)-$square tilings. Call a \textbf{horizontal block} a rectangle of size $sl \times (m+s)$ where $l$ is the least positive integer such that $m$ divides $sl$. There are exactly two ways to tile this block. A \textbf{horizontal flip} is a flip in a horizontal block is a change from one tiling to another. In the same way, define a \textbf{vertical block} as a rectangle of size $(m+s) \times sl$ and let a \textbf{vertical flip} be a flip in a vertical block.
Call a \textbf{central block} a square of size $p \times p$, where $l$ is the least positive integer that divides $m$ and $s$. For $(m,s)-$tilings, where $m,s>1$, there exist exactly two kinds of tilings of the block: by $(p/m)^2$ squares of size $m$ and by $(p/s)^2$ squares of size $s$. A \textbf{central flip} is a flip in a central block. If $(m,s) = 1$, then $p = ms$. For $(1,s)-$tilings there are three tilings of the horizontal/vertical blocks. See Figures \ref{fig: flip12}, \ref{fig: flip23} for an example of flips for the $(1,2)$ and $(2,3)$ tilings.

\begin{figure}
\centering
\includegraphics[scale=0.5]{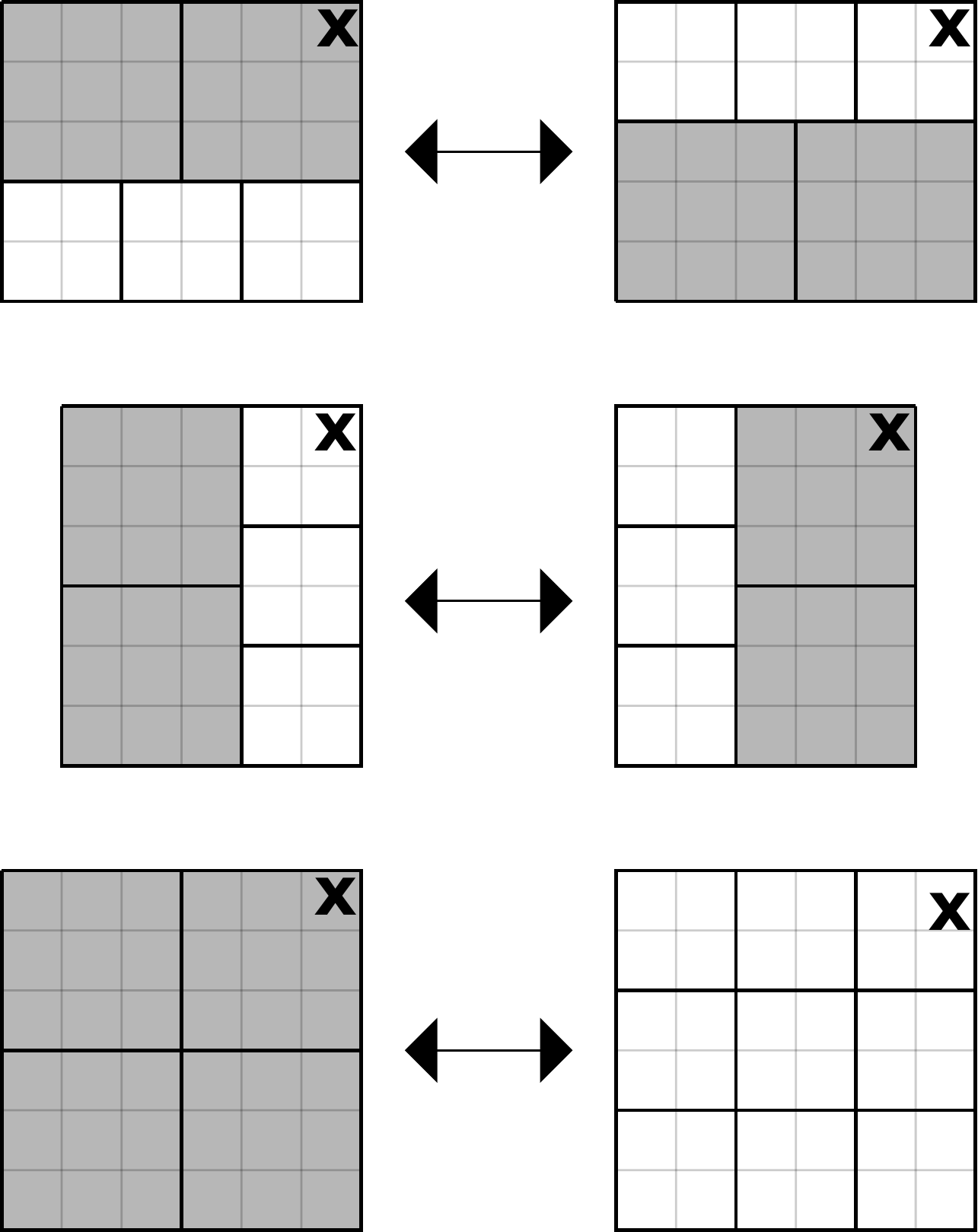}
\caption{Flips for the $(2,3)-$ tiling.}
\label{fig: flip23}
\end{figure}

It was shown in \cite{BR02} that the tiling space is connected by flips. Moreover, there exists a minimal tiling to which one gets by performing height non-increasing flips for every simply connected region. There is also an algorithm that runs in quadratic time over the size of the region that checks tileability by trying to construct a minimal tiling $T_{min}$. This minimal tiling might not be unique. There might be a subset of height equivalent minimal tilings $M = {T_{min}}$. In this case, one can consider a function on the subset of height equivalent tilings (called \textbf{potential} in \cite{BR02}). For $T_{min} \in M$ the potential adds up vertical coordinates of horizontal sides of the first type of tiles and horizontal coordinates of vertical sides of the second type of tiles. $T_{min}$ that has the minimal potential is then the unique ``global'' minimal tiling on the set of tilings $\Omega_R$.

Since there are exactly two tilings for every type of block, we say that each flip has two directions. One way to think about is the following: if a flip changes the height of the tiling, then the direction of the flip is ``up'' if it increases the height and ``down'' if it decreases the height. If the flip does not change the height, then it changes the potential. Let us then say that the direction of the flip is ``up'' if it increases the potential and ``down'' if it decreases the potential.

\subsection{Examples}

\subsubsection*{ $(1,s)-$square tilings}

\begin{figure}
\begin{center}
\includegraphics[scale=0.5]{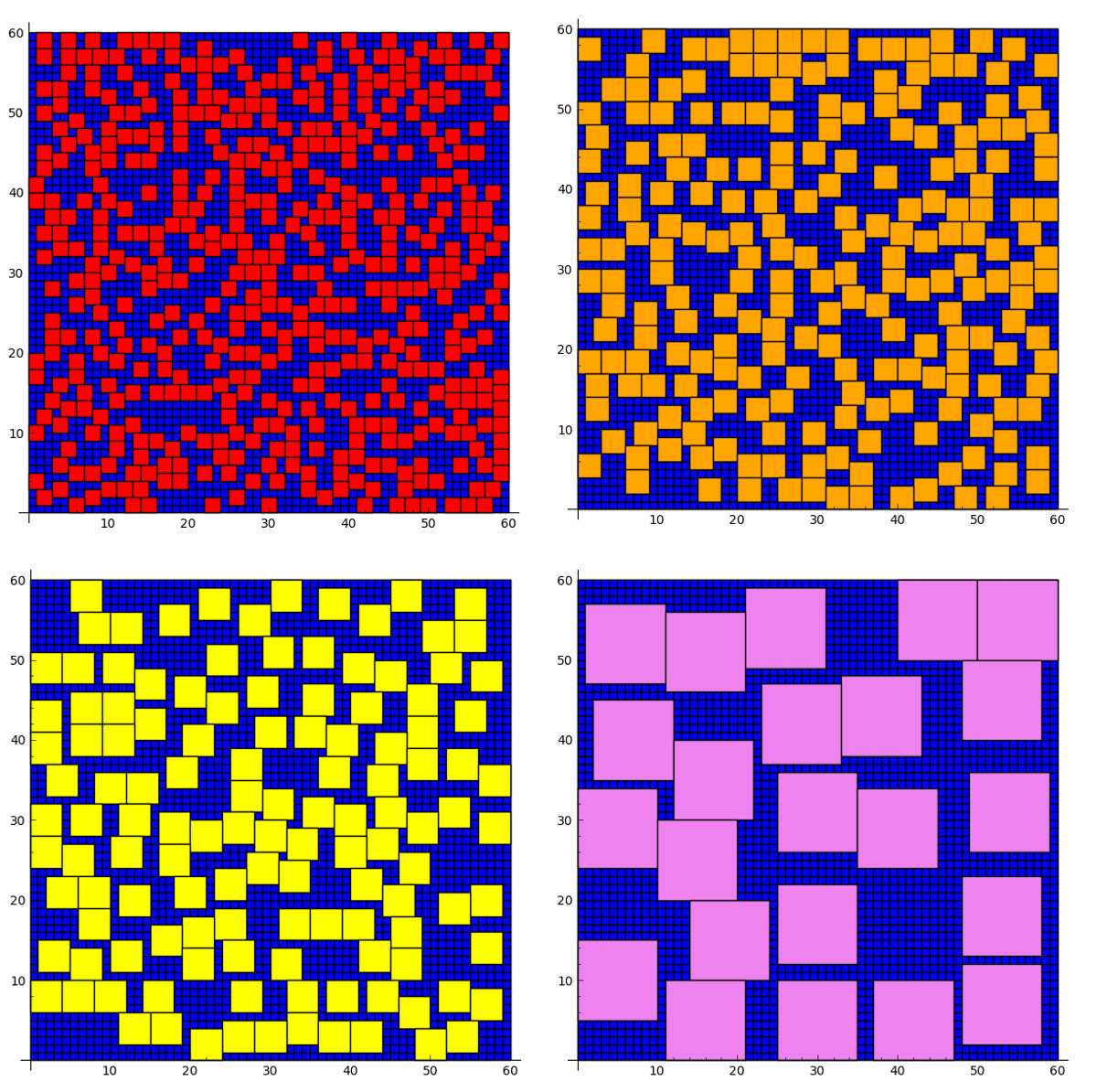}
\end{center}
\caption{ $(1,2), (1,3), (1,4), (1,10)$ tilings of a $60 \times 60$ square. }
\label{fig: squares}
\end{figure}

Consider tilings by $1\times 1$ and $s \times s $ squares of a finite region $R$ of $\mathbb{Z}^2$ of area $N_R$.
The height function defined above is not much of a use in this case, it is simply a constant.

Consider central flips on the set of tilings of a region $R$. One can define a different order on the set of tilings. For example, for a given tiling $T$ let its height $h(T) $ be equal to the minimal number of flips needed to get to from the tiling with only small squares $T_0$ to $T$, set $h(T_0) = 0$. The height function becomes simply a Hamming distance function on the tiling graph $G_{1,s}(R)$. In this setting, $T_0$ is the unique minimal tiling. If the region $R$ can be tiled by big squares only, then the maximal tiling is the tiling by big squares only (speaking about rectangular regions, both of its sides need to have a multiple of $s$ as a length). 

If we consider central flips on the set of tilings of a region $R$, then the diameter of the tiling graph is $d(G_{1,s}) = O(N_R)$. Let $\alpha_{s}$ be the number of squares of size $1$ in a tiling of $R$, then $\alpha_{s} = const (N_R)\mod s^2$. It does not depend on a tiling and holds for any finite region $R$ of $\mathbb{Z}^2$.

Figure \ref{fig: squares} shows examples of $(1,2), (1,3), (1,4), (1,10)$ tilings of a $60 \times 60 $ region. The tiling by $1 \times 1$ and $2 \times 2$ squares (in the upper left corner of Figure \ref{fig: squares}) is a random tiling obtained via coupling (see Subsection \ref{subsec: simulations} for details on simulations). The other three tilings are outputs after $100,000,000$ flips (starting from configurations with only $1 \times 1$ squares).

\subsubsection*{ $(2,3)-$square tilings}

\begin{figure}
\begin{center}
\includegraphics[scale=0.4]{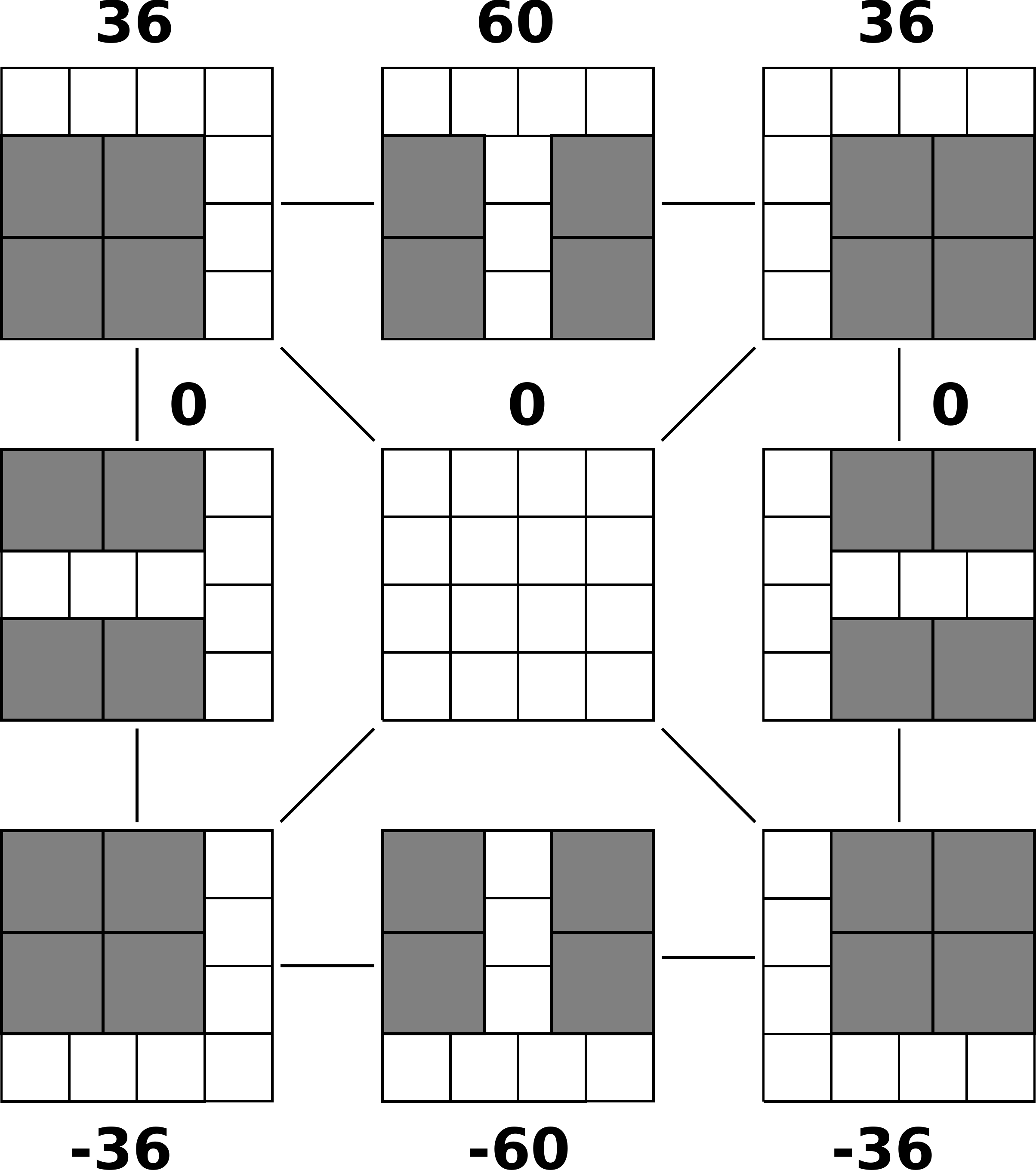}
\end{center}
\caption{All $(2,3)-$square tilings of an $8 \times 8$ square with their corresponding heights.}
\label{fig: tiling_graph_2_3}
\end{figure}

$(2,3)-$square tilings (and $(m,s)$ in general) are completely different from the $(1,s)$ case. It is no longer possible to glue two parts of a tiling together that easily. It depends a lot on the boundary of the region. The height function is linear over the size of the region. Figure \ref{fig: tiling_graph_2_3} shows all possible $(2,3)-$square tilings of a $8 \times 8$ region where two configurations are connected if a flip can be made to go from one to the other.


\subsection{Markov chain}
\label{subsec: Markov chain}

Let $R$ be a finite simply connected region of area $N$ of $\mathbb{Z}^2$. If $R$ can be tiled by squares of sizes $m,s$, denote the set of all possible tilings by $ \Omega_R $ (we omit $m,s$ for simplicity). Let us define a Markov chain \textbf{MC$_{square}$} for square tilings. The idea of the chain is to pick a site of the region uniformly at random at each step and do a flip if possible. The little black crosses in the Figures \ref{fig: flip12}, \ref{fig: flip23} mark the site that has to be chosen in order to perform a flip. 

First of all, for $(m,s)-$tilings, $m,s>1$, the site that has to be chosen to perform a flip is always in the upper right corner of the block (there are horizontal/vertical and central blocks in which a flip can be performed, as it was defined in Subsection \ref{subsec: flips}). One can see from Figure \ref{fig: flip23} that once a site of a tiling is chosen, there is no ambiguity in what type of flips can be performed: in order to perform a flip, one has to consider the three blocks in which this site is in the upper right corner, and there is at most one type of flips that can be performed. If $m=1$, 
one can see from Figure \ref{fig: flip12} that there is no ambiguity about what kind of flip has to be performed for $(1,s)-$tilings in a horizontal/vertical block: if one wishes to perform a central flip, then the upper right corner of the big  square has to be chosen, if one wishes to perform a horizontal flip, then one has to choose the site that will correspond to the upper right corner of the big square after it is moved.

Second of all, in order to choose one of the two configurations of the block (in the case $m=1$, the choice nails down to choosing in which direction we want to push the big square), let us recall that each flip has two directions, so before performing a flip we first choose one of the two possible directions. 

Let us now formally present the Markov chain:

\subsection*{MC$_{square}$:}

Let $T_0 \in \Omega_R$ be an initial configuration. At each time $t$:
\begin{itemize}
\item choose an inner vertex of $R$ u.a.r.,
\item choose a direction of the flip with equal probability,
\item perform either a vertical/horizontal or central flip in the tiling $T_t$ in the chosen direction if possible thus defining the tiling $T_{t+1}$, otherwise stay still.
\end{itemize}

\begin{lemma}
MC$_{square}$ has uniform stationary distribution.
\end{lemma}

\begin{proof}

The probability to reach every tiling $T \in \Omega_R$ is positive since the state space is connected, so MC$_{square}$ is irreducible. The probability to stay in the same state is positive, therefore it is aperiodic. This implies that the chain is ergodic and thus has unique stationary distribution. Moreover, the probability matrix of MC$_{square}$ is symmetric: indeed, $\mathbb{P}(T,S) = \mathbb{P}(S,T)$ for all pairs $(T,S)$ from $\Omega_R$ (that are different by one flip). The symmetry of the probability matrix ensures the uniformity of the stationary distribution. 
\end{proof}

\section{Mixing time for $(1,s)-$square tilings}
\label{subsec: fast mixing}

The question of proving fast mixing turns out to be difficult. There are only few tiling systems for which the dynamics were proven to be fast mixing (for example, \cite{LRS95, W04}), where by fast mixing we mean in polynomial number of steps in the size of the tiled region. We show that in some particular cases, the dynamics is rather easy to understand and is related to known structures such as independent sets.


\begin{theorem}\label{thm: mixing time for nxlong }
MC$_{square}$ is rapidly mixing for $(1,s)-$square tilings of a rectangular region of size $ n \times \log  n $.
\end{theorem}

\begin{figure}
\centering
\includegraphics[scale=0.55]{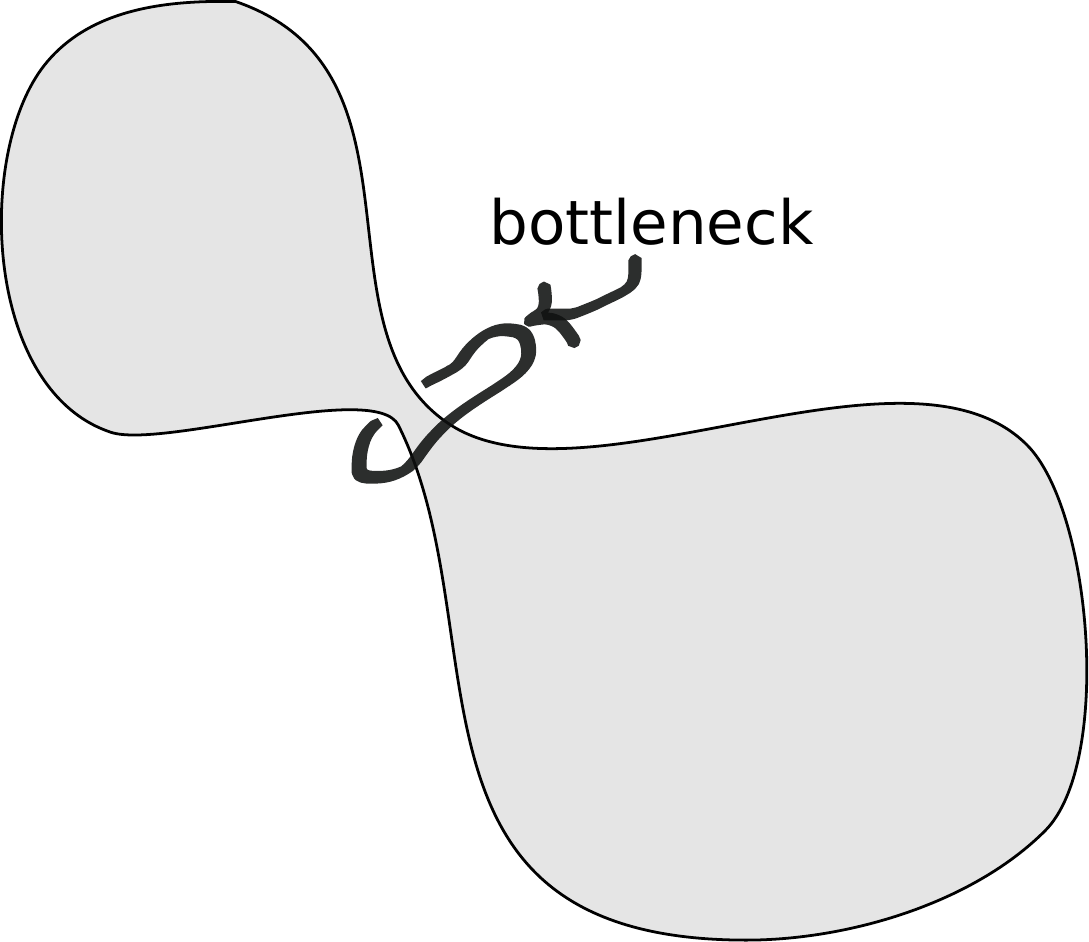}
\caption{A graph with a bottleneck.}
\label{fig: bottleneck}
\end{figure}

\begin{proof}

We are going to use the \textbf{canonical paths} argument developed by Sinclair \cite{S92}. The idea of canonical paths is the following: fast mixing occurs when a tiling graph does not have a \textbf{bottleneck}. A bottleneck is a geometric feature of the state space of a MC that controls mixing time (see Figure \ref{fig: bottleneck} for a sketch of a graph with a botteleck). If there is a bottleneck, then it divides the set of states of the MC (in our case, tiling configurations from $\Omega$) into two subsets connected by a thin ``tunnel''. This slows down the mixing as it becomes hard to get from one subset to the other. Canonical paths allow to formalize absence (or presence) of a bottleneck. For each pair of configurations a canonical path or a set of paths are defined, that allow to get from one tiling to the other via flips. Consider an edge in the tiling graph (it connects two configurations different by one flip). If for any edge the number of paths that pass through this edge is relatively small (linear over the cardinality of the set of configurations), the graph does not have a bottleneck. 

Let us present the construction for the $(1,2)$ case. It is exactly the same for the general case. Place the rectangle on the grid with the lower left corner in $(0,0)$. Define a lexicographic order on the set of inner vertices of the $(n+1)\times (\log  n+ 1)$ rectangle: $$ \left\lbrace (1,1), (1,2), \ldots, (1, \log n),(2,1), (2,2),  \ldots, (n, \log n - 1), (n, \log n) \right\rbrace $$  and denote them as $ \left\lbrace 1, 2, \ldots \tilde{n} \right\rbrace $, where $ \tilde{n} = n \times \log n$. 

Consider two tilings $X$ and $Y$ from $\Omega$. In order to get from $X$ to $Y$, it is sufficient to make a flip in every vertex of the region, in other words, the diameter of the tiling graph is not greater than the area of the region. With the use of canonical paths, we can structure the order in which we perform flips. Define a canonical path $p(X,Y)$ from $X$ to $Y$ as follows: start from $X$, follow the vertices in the defined order in windows of size $2 \times 2$ and in each window perform at most one flip in each vertex if it decreases the flip-distance (Hamming distance in the tiling graph) to $Y$. Such ordering of flips at each step either corrects a $2 \times 2$ box in $X$ by performing a central flip or decreases distance by one by performing a vertical/horizontal flip in a vertical/horizontal block. After all vertices are met once, the path reaches  $Y$. Due to the construction, for any pair of tilings such path exists and is unique. The ordered path is now a permutation $\sigma$ of vertices $1 \ldots \tilde{n}$, where for each $i:\;\;$ $  i - \log n - 1 \leq \sigma(i) \leq i + \log n + 1$. 

\begin{figure}
\centering
\includegraphics[scale=0.55]{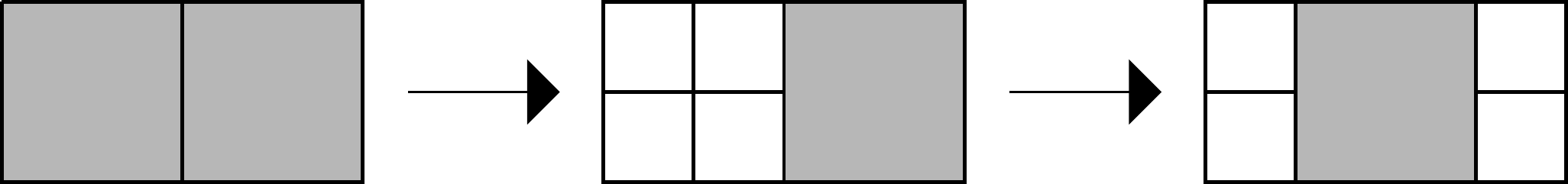}
\caption{Canonical path for $(1,2)-$ tilings in a $4\times 2$ region.}
\label{fig: canonical}
\end{figure}

Denote by $\pi$ the chain's stationary distribution. Send $\pi(X)\pi(Y)$ units of \textbf{flow} through $p(X,Y)$ for any pair $(X,Y)$. The flow through an edge is just the sum of all the flow that travels through the edge. The idea of the canonical paths method is to prove that any edge in the tiling graph has a small number of canonical paths going through it and therefore little flow. Consider an edge $e_i = (Z_{\sigma(i)}, Z_{\sigma(i+1)})$, where tilings $Z_{\sigma(i)}$ and $Z_{\sigma(i+1)}$ only differ in the $\sigma(i)-$th position.

\begin{lemma}
\label{lemma: canonical}
There are no more than $\vert \Omega \vert n $ paths passing through each $e_i$.
\end{lemma}

\begin{proof}
Consider the canonical path $p$ going from $X$ to $Y$:
$$
p(X,Y) = \lbrace X = Z_{\sigma(1)}, \ldots, Z_{\sigma(i)}, Z_{\sigma(i+1)}, \ldots Z_{\sigma(\tilde{n})}= Y  \rbrace.
$$

Tiling $X$ agrees with $Z_{\sigma(i)}$ at least on the vertices $\sigma(i+1), \ldots, \sigma(\tilde{n})$. $Y$ agrees with $Z_{i}$ at least on the vertices $\sigma(1), \ldots, \sigma(i)$. One could think of reconstructing $X$ and $Y$ using $e_{i}$ and the first $i$ vertices $1, \ldots, \sigma(i)$ of $X$ at the last $\tilde{n} - i$ vertices $\sigma(i), \ldots \sigma(\tilde{n})$ of $Y$, but that would mean that it was possible to glue two pieces of a tiling with holes together. Consider instead a strip of width $4$ around the windows that contain the vertex $i$ between the two part of tilings which can be filled in at most $exp(\log n)$ ways. 

We have therefore just constructed a map from the set of paths $\lbrace p_{e_i} \rbrace$ that pass through $e_i$ to the state space $\Omega$. This construction maps each path $p_{e_i}$ to $n$ tilings that are different only in the $ 4 \times \log n$ strip around the vertex $i$. Each path is mapped to a different ``family'' of $(1,s)-$ tilings, where each family corresponds to tilings of the strip around a given vertex. There are not more than $\vert \Omega \vert $ possible families and therefore not more than $\vert \Omega \vert n $ paths through $e_i$.

\end{proof}

Let us continue with the proof of the theorem. Applying Lemma \ref{lemma: canonical}, the flow along each $e_i$ is at most $\vert \Omega \vert n$.

The \textbf{cost} of the flow $f$ is 
$$
cost(f) := \max _{e_i} \frac{\vert \Omega \vert n \pi(X) \pi(Y)}{c(e_i)},
$$
where  $X$ and $Y$ are the endpoints of the paths that go through $e$, $c(e_i)$ is the \textbf{edge capacity}:
$$
c(e_i): = \pi(Z_i) \mathbb{P}(Z_i, Z_{i+1}).
$$

Since 
\begin{equation}
\label{eq: pi for sq}
\pi(X)  = \pi(Y) = \frac{1}{\vert \Omega \vert}, 
\end{equation}
one gets the following bound on the cost:

\begin{equation}
\label{eq: cost}
cost(f) \leq 2n\tilde{n} = 2n^2\log n.
\end{equation}

There is the following relation between the cost function and the mixing time (see \cite{S92}, Proposition 1 and Corollary 6'):

\begin{equation}
\label{eq: cost for mixing}
\tau_{mix}(\varepsilon) \leq 8cost^2 (\ln \pi(X)^{-1} + \ln\varepsilon^{-1}).
\end{equation}

Plugging \eqref{eq: pi for sq} and \eqref{eq: cost} into \eqref{eq: cost for mixing}, the following bound is obtained on the mixing time of MC$_{square}$:

\begin{equation}\label{eq: nxlong mixing bound}
\tau_{mix}(\varepsilon)\leq 32 n^4\log^2 n (c_2 n \log n + \ln\varepsilon^{-1}).
\end{equation}

The same reasoning works for any $(1,s)-$ square tiling -- the vertical strip has to be taken of length $const(s)$. And  $\tau_{mix}$ stays polynomial.
\end{proof}

\textbf{Remark 1.} Theorem \ref{thm: mixing time for nxlong } works not only for rectangular regions but for any regions, such that in any site the region can be divided in two parts via a strip of width $const(s)$ and height $\log n$.

\textbf{Remark 2.} Simulations via coupling (see Section \ref{subsec: simulations} for the description of the algorithm) suggest the $O(n^2)$ bound for the $(1,2)$ case, which shows that the bound obtained in \eqref{eq: nxlong mixing bound}  is not optimal. Let us remind that the evident lower bound is of the size of the diameter of the tiling graph which is simply the area of the region. 

\section{Weighted Glauber dynamics}
\label{subsec: weighted}

It is not clear how to prove fast mixing in the general case, so in this part let us consider a weighted version of the dynamics for the square tilings. It simply means that we favorize some configurations more than others. In the case of $(1,s)-$tilings it seems that the big squares significantly slow down the mixing time, so a way to go around it is to put less probability weight on the big squares. A Markov chain associated with the system that does local transformations, e.g. flips, (whose stationary distribution is the desired distribution) is generally referred to as the \textbf{Glauber dynamics}. It is popular in statistical physics and used to describe different behaviours of systems (e.g., Ising model, Hardcore model, independent sets, perfect matchings, etc.). Weights $\lambda$ assigned to the particles usually correspond to the energy. Weights can help to detect presence of \textbf{phase transitions} in these systems -- there exists a critical point $\lambda_c$ such that the dynamics is fast mixing (in polynomial time over the size of the problem) below this critical point, for all $\lambda < \lambda_c$, and is slow mixing (in exponential time) for all  $\lambda > \lambda_c$. It is usually difficult to understand what is happening in the critical point. Let us just point out that that for a variety of studied models it has not been possible to analyze the behaviour of Markov chains at critical points and sometimes in their neighbourhoods.   

It seems that for $(1,s)-$tilings $\lambda = 1$ is the critical point. $\lambda = 1$ corresponds to the unweighted version of the chain. We consider the $(1,s)-$case and prove fast mixing for certain $\lambda \leq \frac{1}{(2s-1)^2-2}$. In the case $s=2$, the conditional bound on $\lambda$ is better because of the relation to the independent sets. 

\subsection{$(1,2)-$square tilings with weights}

Consider $(1,2)-$ square tilings as the King's problem on a toroidal region of $\mathbb{Z}_2$. The King's problem is a problem of placing non-attacking kings on a chessboard: if a king is placed in the site, none of 8 neighbouring sites can be occupied. It can be seen as an independent set problem on the $8-$adjacency graph on the square grid graph $G = (V,E)$ of degree $8$ (see Figure \ref{fig: 8_regular}). An independent set of $G$ is a subset of vertices such that no two of them are adjacent (for more information about independent sets see, e.g., \cite{DG00,DFJ02}).

\begin{figure}
\centering
\includegraphics[scale=0.7]{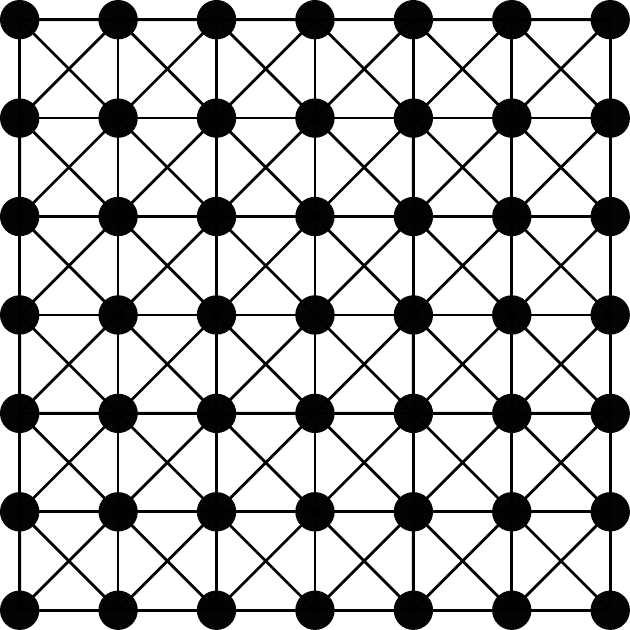}
\caption{8-regular graph on the square grid.}
\label{fig: 8_regular}
\end{figure}

Consider the Glauber dynamics of this system. Let $I(G)$ be the set of all independent sets of $G$. The probability of a configuration $X$ is given by 
$$
\pi(X) := \frac{\lambda^{\vert X \vert}}{Z(\lambda)},
$$
where $\lambda$ is a positive parameter called a \textbf{weight} of a configuration and $Z(\lambda) -$ the \textbf{partition function} of the system:
$$
Z(\lambda) := \sum _{X \in I(G)  } \lambda^{\vert X \vert}.
$$

Define the weighted version of MC$_{(1,2)-square}$ as follows.  We add a \textbf{diagonal dragging} flip which is shown in Figure \ref{fig: drag}.

\subsection*{MC$_{(1,2)-square}$ with $\lambda > 0$:}
Start from $X_0$. Let $X_t$ be the configuration at time $t$. At time $t$:
\begin{itemize}
\item Choose a site of the region u.a.r.,
\item If a central flip can be made, put four $1\times1$ squares with probability $\frac{1}{\lambda +1 } $, put a $2\times2$ square with probability $\frac{\lambda}{\lambda +1 } $. Else, if a horizontal/vertical/diagonal (dragging) flip can be made, perform it with probability $\frac{\lambda}{4(\lambda +1)}.$
\item Otherwise, do nothing and set $X_{t+1} = X_{t}$.
\end{itemize}

\begin{figure}
\centering
\includegraphics[scale=0.5]{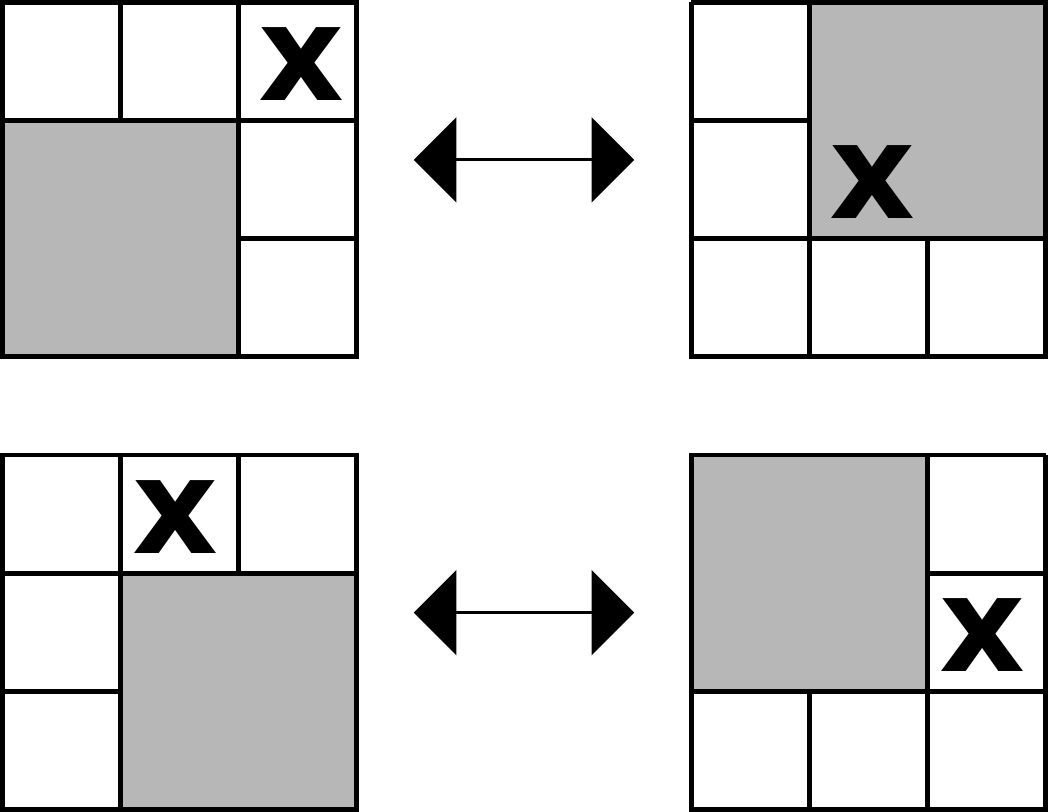}
\caption{A diagonal drag flip for $(1,2)-$ tiling. Each cross marks the position of the upper right corner of the big square after the flip.}
\label{fig: drag}
\end{figure}

This version MC$_{(1,2)-square}$ corresponds exactly to the \textit{delete/insert/drag} chain MC$_{drag}$ for the independent sets \cite{DG00}, which is a slightly different version of the Luby-Vigoda chain \cite{LV99}:

\subsection*{MC$_{drag}$  for independent sets:}

Start from $X_0$. Let $X_t$ be the configuration at time $t$. At time $t$:
\begin{itemize}
\item Choose $v$ from the set of vertices u.a.r.,
\item If $v \in X_t$, then \textbf{delete} $v$ : $X_{t+1} = X_{t} \setminus \lbrace v \rbrace$ with probability $\frac{1}{\lambda +1} $,
\item If $v \notin X_t$, then \textbf{add} $v$ : $X_{t+1} = X_{t} \cup \lbrace v \rbrace$ with probability $\frac{\lambda}{\lambda +1} $,
\item If $v \notin X_t$ and $v$ has a unique neighbour $u$ in $X_t$, then \textbf{drag} $v$ : $X_{t+1} = (X_{t} \cup \lbrace v \rbrace) \setminus \lbrace u \rbrace$ with probability $\frac{\lambda}{4(\lambda +1)} $,
\item Otherwise, do nothing and set $X_{t+1} = X_t$.
\end{itemize}

The two main tasks are to approximately evaluate the partition function $Z(\lambda)$ and to approximately sample from $I(G)$ according to the stationary distribution. When the graph's maximal degree $\Delta$ is greater than $2$, approximate evaluation of $Z(\lambda)$  and approximate sampling from $I(G)$ can be done using a rapidly mixing chain (see, for example, \cite{JS96}). Using the path coupling argument (Theorem \ref{theorem: DG} in Preliminaries), Dyer and Greenhill proved fast mixing for MC$_{drag}$ with sufficiently small $\lambda$. Rapid mixing for MC$_{(1,2)-square}$ then follows directly when $\lambda \leq \frac{1}{3}$. 
 
\begin{theorem}[Dyer-Greenhill \cite{DG00}]
\label{thm: DG independent}
Let $G=(V,E)$ be a graph with maximal degree $\Delta$ and $ \vert V \vert = n$. MC$_{drag}$ is rapidly mixing for $\lambda \leq 2/(\Delta - 2)$. 
\begin{enumerate}
\item When $\lambda < 2/(\Delta - 2)$, 

$$
\tau_{mixing}(\varepsilon) \leq \frac{2(1+\lambda)}{2-(\Delta-2)\lambda} n\log (n\varepsilon^{-1}).
$$
\item When $\lambda = 2/(\Delta - 2)$, 
$$
\tau_{mixing}(\varepsilon) \leq \lceil 2n^2(1+ \lambda)(\log n + 1)\rceil \lceil \log \varepsilon^{-1}\rceil.
$$
\end{enumerate}
\end{theorem}

Using Dyer-Greenhill's theorem, we get the following bound on the mixing time for MC$_{(1,2)-square}$:

\begin{theorem}
\label{th: weighted mixing}
Consider $(1,2)-$square tilings of an $n \times n $ toroidal region. 
MC$_{(1,2)-square}$ is rapidly mixing for $\lambda \leq \frac{1}{3}$. The following bounds stand for $\tau_{mix}(\varepsilon)$ and some positive constant $c_1, c_2$.

\begin{enumerate}
\item When $\lambda < 1/3$, 

$$
\tau_{mix}(\varepsilon) \leq c_1  n^2 \log (n\varepsilon^{-1}).
$$

\item When $\lambda = 1/3$, 
$$
\tau_{mix}(\varepsilon) \leq c_2 n^4 \lceil (\log n + 1) \log \varepsilon^{-1}\rceil.
$$
\end{enumerate}

\end{theorem}

\subsection{$(1,s)-$square tilings with weights}

Consider now the weighted dynamics for the $(1,s)$ case. Weights are put on the squares of size $s$ (big squares). Consider the natural Markov chain with only central flips. 

\subsection*{MC$_{(1,s)-square}$ with $\lambda > 0$:}
Start from $X_0$. Let $X_t$ be the configuration at time $t$. At time $t$:
\begin{itemize}
\item Choose a site of the region u.a.r.,
\item If a central flip can be made -- put $s^2$ $1 \times 1$ squares with probability $\frac{1}{\lambda +1} $ or an $ s \times s$ square with probability $\frac{\lambda}{\lambda +1} $, thus defining $X_{t+1}$,
\item Otherwise do nothing and set $X_{t+1} = X_{t}$.
\end{itemize}

For two tilings $A$ and $B$ of the region $R$ by $1 \times 1$ and $s \times s $ flips let $\varphi(A,B)$ denote the minimal number of flips one has to perform to get from $A$ to $B$. It follows directly from the definition of a flip that $\varphi(A,B) = \varphi(B,A)$ for any $A,B$.

It turns out that with $\lambda$ sufficiently small, the coupling time for MC$_{(1,s)-square}$ is polynomial. Namely, we get the following result.

\begin{theorem}
\label{thm: $(1,s)$ weighted coupling}
Consider $(1,s)-$square tilings of an $n \times n $ region. 
MC$_{(1,s)-square}$ is rapidly mixing for $\lambda \leq \frac{1}{(2s-1)^2-2}$ and there exists a positive constant $c$ such that
$$
\tau_{mix}(\varepsilon) \leq c n^4 \log (n\varepsilon^{-1}).
$$
\end{theorem}

\begin{proof}
Consider a coupling  $(A_t,B_t)_t$ and two configurations  $A_t$ and $B_t$ at time $t$ that are different by one flip, thus $\varphi(A_t,B_t) =1$. We want to apply the coupling theorem by Dyer and Greenhill  and prove that $\mathbb{E} \Delta \varphi \leq 0$. 

Consider $s = 2$. In the worst case scenario, there are 8 bad sites that increase the distance between the two configurations and only one that decreases. A bad flip always implies putting a big square in the configuration with small squares. It is done with probability $\frac{\lambda}{N(\lambda  + 1)}$, where $N  = n^2$ is the area of the region. The one good site decreases the distance for both direction of a flip: this is done with  probability $\frac{\lambda}{N(\lambda  + 1)} + \frac{1}{N(\lambda  + 1)}$. So
 
\begin{equation}\label{eq: thm $(1,s)$ weighted coupling 1 }
\mathbb{E}(\Delta \varphi) \leq -\frac{1}{N} +  \frac{8 \lambda}{N(\lambda  + 1)}.
\end{equation} 

This means that when the right part of \eqref{eq: thm $(1,s)$ weighted coupling 1 } is  not greater than $0$, the chain is rapidly mixing. By solving the inequality one gets the condition on $\lambda$: $\lambda \leq \frac{1}{7}$. 

When $s > 2$, the number of bad sites does not exceed $(2(s-1)+1)^2 - 1$, so:

\begin{equation}\label{eq: thm $(1,s)$ weighted coupling 2 }
\mathbb{E}(\Delta \varphi) \leq -\frac{1}{N} +  \frac{((2s-1)^2) -1 ) \lambda}{N(\lambda  + 1)}.
\end{equation} 

$ \mathbb{E}(\Delta \varphi) \leq 0 $ whenever $ \frac{((2s-1)^2 -1 ) \lambda}{(\lambda  + 1)}\leq 1. $
This is true when 
$$
\lambda \leq  \frac{1}{(2s-1)^2) - 2}.
$$

In order to apply the coupling theorem, we also need there to exist $\alpha > 0$ such that $\mathbb{P}(\varphi(A_{t+1},B_{t+1}) \neq  \varphi(A_t,  B_t)) \geq \alpha $. The inequality holds  for $\alpha = \frac{1}{N}$. 

Now we can freely apply the theorem and get the following bound on the mixing time of an $n \times n $ region:

$$
\tau_{mix} \leq  c D^2 n^2,
$$
where $c$ is some positive constant,  $D$ is the diameter of the tiling graph. Since $D$ is $O(n^2)$, one gets the desired bound on  $\tau_{mix}$.

\end{proof}

\textbf{Remark 1.} Simulations with weight parameter $\lambda$ from Theorem \ref{thm: $(1,s)$ weighted coupling} suggest coupling in $O(n^2)$ steps.

\textbf{Remark 2.} We considered only central flips in the above theorem. Mixing time will stay polynomial if one considers horizontal/vertical/dragging flips as well. But it makes the calculations more cumbersome.

\section{Simulations}
\label{subsec: simulations}

Let us describe the algorithm used for simulations. We use Python language to run the simulations and Sage graphics for the pictures.\\

The algorithm \textbf{1} describes a basic coupling approach for getting a sample of a $(m,s)-$square tiling of a given regions. Since the chain is not monotone, we cannot use Coupling From the Past \cite{LPW09}. The main problem though is which configurations to choose as initial configurations of the coupling. The initial configurations are chosen to be a pair of tilings which are as far as possible from each other in the tiling graph. We choose a pair of a minimal and maximal tilings, knowing that they need the maximal number of flips to be performed in order to get from one to the other. The number of steps of the algorithm after which they meet gives an estimate on the coupling time and an idea on the general look of the tiling.\\

\begin{algorithm}
\KwData{$n$, $m$ and $s$}
\KwResult{Random tiling of a $n\times n$-square by $m\times m$ and $s\times s$ square tiles}
$A\gets$ max tiling\;
$B\gets$ min tiling\;
$t\gets 0$\;
\While{$A\neq B$}
{
direction$\gets$ random choice up or down\;
position $\gets$ random choice of vertex\;
$A\gets $flip($A$, direction, position)\;
$B\gets $flip($B$, direction, position)\;
$t\gets t+1$\;
}
\label{algo:1}
\caption{Uniform sampling by coupling}
\end{algorithm}

 Table \ref{tab:avg} shows estimates on the average coupling time over 100 trials for $(1,s)-$tilings of a $n \times n $ region for $s=2, \ldots,10$ and $n = 10,\ldots, 30$. See Figure \ref{fig: squares_coupling} for examples of $(4,7)$ and $(3,10)-$tilings obtained by coupling.\\

\begin{table}

\resizebox{0.7\textwidth}{!}{\begin{minipage}{\textwidth}

\begin{center}
 \begin{tabular}{|r|c|c|c|c|c|c|c|c|c|}
\hline 
$n \setminus s$ & 2 & 3 & 4 & 5 & 6 & 7 & 8 & 9 & 10 \\
\hline
10 & $1.2 \times 10^3$ & $2.7 \times 10^3$&   $1.01 \times 10^3$ &$4 \times 10^2$ & $ 2\times 10^2$  & $ 1.6 \times 10^2$ & $3 \times 10^1 $& $1 \times 10^1 $ & $1$    \\
11 & $2.9 \times 11^3$ & $5.02 \times 11^3$  &  $2.  \times 11^3$ &$1.01 \times 11^3$ & $ 1.3\times 11^3$  & $ 1.3 \times 11^2$ & $2 \times 10^2 $ &$1.3 \times 10^{1.5} $ & $0.8 \times 10^1 $ \\
12 & $2.2 \times 12^{3.5}$ & $0.8 \times 12^4$&  $3.1 \times 12^3$ &$2.1 \times 12^3$ & $ 1.3 \times 12^3$  & $ 0.9 \times 12^3$ & $3 \times 10^2 $ & $2.5 \times 10^2 $&  $2.1 \times 10^{1.5} $\\
13 & $2.3 \times 13^{3.5}$  & $1.1 \times 13^4$&  $2.6 \times 13^{3.5}$ &$3.1 \times 13^3$ & $ 1.1\times 13^3$  & $ 1.24 \times 13^3$ & $1.169 \times 13^{2.5} $ &   $1.7 \times 13^{2.5} $& $1.01 \times 13^2 $\\
14 & $1.5 \times 14^{3.5}$  & $1.2 \times 14^4$&  $0.9 \times 14^4$ &$5.1 \times 14^3$ & $ 2.3 \times 14^3$  & $ 1.01 \times 14^2$ & $0.9 \times 14^3 $& $0.83 \times 14^3 $&  $1.73 \times 14^2 $\\
15 & $1.1 \times 15^{3.5}$  & $1.9 \times 15^4$& $1.2 \times 15^4$  &$8.2 \times 15^3$ & $ 4.1 \times 15^3$  & $ 1.1 \times 15^3$ & $1.7 \times 15^3$&  $1.3 \times 15^3$& $1.9 \times 15^{2.5}$\\
16 & $1.1 \times 16^{3.5}$  & $2.5 \times 16^4$& $1.9 \times 16^4$  &$0.9 \times 16^4$ & $ 1.01 \times 16^{3.5}$  & $ 0.8 \times 16^3$ & $1.95 \times 16^3$&$2.2 \times 16^3$ & $3.9 \times 16^{2.5}$\\
17 & $2.5 \times 17^{3.5}$  & $3.01 \times 17^4$& $2.9 \times 17^4$  &$0.9 \times 17^4$ & $ 2.5 \times 17^{3.5}$  & $ 3.1 \times 17^3$ & $ 2.6 \times 17^3$& $ 2.82 \times 17^3$& $ 0.95 \times 17^3$\\
18 & $2.4 \times 18^{3.5}$  & $4.9 \times 18^4$&  $1.01 \times 18^{4.5}$ &$1.3 \times 18^4$ & $ 0.8 \times 18^{4}$  & $ 1.9 \times 18^3$ & $ 3.5 \times 18^3$& $ 3.3 \times 18^3$ & $ 1.7 \times 18^3$\\
19 & $2.9 \times 19^{3.5}$  & $5.1 \times 19^4$&   $1.4 \times 19^{4.5}$ &$2.01 \times 19^4$ & $ 0.85 \times 19^{4}$  & $ 2.5 \times 19^3$ & $ 1.1 \times 19^{3.5}$& $ 0.89 \times 19^{3.5}$& $ 2.2 \times 19^3$\\
20 & $2.9 \times 20^{3.5}$  & $1.01 \times 20^4$&  $2.9 \times 20^{4.5}$ &$3.7 \times 20^4$ & $ 1.1 \times 20^{4}$  & $ 1.2 \times 20^3$ & $1.18 \times 20^{3.5}$& $1.32 \times 20^{3.5}$& $2.5 \times 20^{3}$\\
21 & $2.9 \times 21^{3.5}$  & $1.4 \times 21^4$& $1.8 \times 21^{4.5}$  &$2.6 \times 21^{4.5}$ & $ 1.6 \times 21^4$  & $ 1.01 \times 21^{3.5}$ & $1.8 \times 21^{3.5}$& $1.7 \times 21^{3.5}$& $0.7 \times 21^{3.5}$\\
22 & $2.3 \times 22^{3.5}$  & $1.8 \times 22^4$&  $1.2 \times 22^{4.5}$ &$2.2 \times 22^{4.5}$ & $3.01 \times 22^4$  & $ 3.2 \times 22^4$ & $3.2 \times 22^{3.5}$& $2.1\times 22^{3.5}$& $0.8 \times 22^{3.5}$\\
23 & $2.9 \times 23^{3.5}$  & $1.2 \times 23^5$& $0.7 \times 23^5$  &$0.8 \times 23^5$ & $ 1.01 \times 23^{4.5}$  & $ 1.02 \times 23^3$ & $1.23 \times 23^{4}$& $2.9 \times 23^{3.5}$& $1.7 \times 23^{3.5}$\\
24 & $2.2 \times 24^{3.5}$  & $0.7 \times 24^5$& $1.8 \times 24^{4.5}$  &$0.6 \times 24^5$ & $ 0.7 \times 24^5$  & $ 1.7 \times 24^{3.5}$ &$1.4 \times 24^4$ & $1.05 \times 24^4$& $2.2 \times 24^{3.5}$\\
25 & $2.9 \times 25^{3.5}$  & $1.4 \times 25^4$&  $2.5 \times 25^{4.5}$ &$0.7 \times 25^5$ & $ 1.01 \times 25^5$  & $ 2.5 \times 25^4$ &  $ 1.8 \times 25^4$& $ 1.79 \times 25^4$& $ 0.77 \times 25^4$\\
26 & $2.3 \times 26^{3.5}$  & $2.7 \times 26^5$&   $1.1 \times 26^5$ &$0.7 \times 26^5$ & $ 1.1 \times 26^5$  & $ 1.8 \times 26^4$ & $ 2.6 \times 26^4$& $ 3.2 \times 26^4$& $ 1.28 \times 26^4$\\
27 & $2.9 \times 27^{3.5}$  & $2.5 \times 27^5$&   $1.5 \times 27^5$&$3.1 \times 27^5$ & $ 1.2 \times 27^{4.5}$  & $ 1.5 \times 27^{4.5}$ & $ 1.02 \times 27^{4.5}$& $ 2.8 \times 27^{4}$& $ 1.2 \times 27^{4}$\\

28 & $2.9 \times 28^{3.5}$  & $3.1 \times 28^5$&  $2.7 \times 28^5$ & $2.3 \times 28^5$ & $ 0.7 \times 28^5$  & $ 1.8 \times 28^{4.5}$ & $ 1.3 \times 28^{4.5}$& $ 1.2 \times 28^{4.5}$& $ 1.7 \times 28^{4}$\\

\hline 
\end{tabular}
\end{center}
 \end{minipage}}

\caption{Average coupling time for $(1,s)$ tilings of $n \times n $ squares over 10 to 100 trials.}
\label{tab:avg}
\end{table}

\begin{figure}
\begin{center}
\includegraphics[scale=0.4]{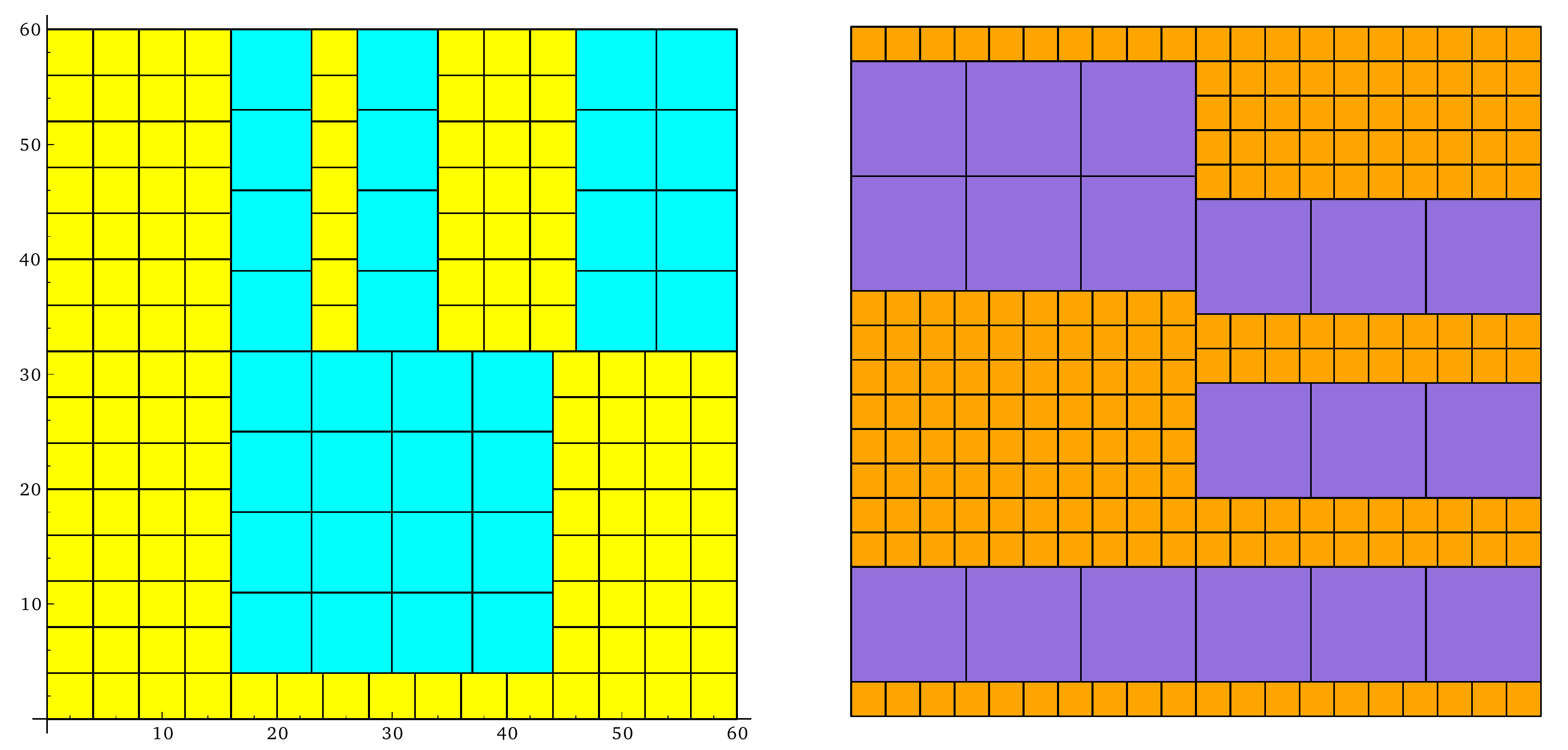}
\end{center}
\caption{ $(4,7)$ and $(3,10)-$square tilings of a  $60 \times  60$ square via coupling.}
\label{fig: squares_coupling}
\end{figure}

\begin{conjecture}
Let $R$ be an $n \times n$ square region of area $n$, tiled by squares of size $m$ and $s$, $m < s$. Then MC$_{(1,2)-square}$ is rapidly mixing and
$$
\tau_{mix}^{(1,2)} = O(n^{3.5})
$$
Moreover,  for $s > 2$ MC$_{(m,s)-square}$ is not rapidly mixing 
its mixing time $\tau_{mix}^{(m,s)} $ is sub-exponential and has the following bound:
$$
\tau_{mix}^{(m,s)}  = \Theta (exp([n/s] + 1)).
$$
\end{conjecture}

\section{Limit shape}
\label{subsec: limit shape}

Let us consider a region such that the height function is not a constant on its boundary but rather grows linearly. One can think of an ``Aztec diamond-type'' region. It seems that in this case, tiling by squares of sizes $m$ and $s$, if both $m,s > 1$ have a typical limiting look, similar to the Arctic circle for dimer tilings (see \cite{CLP02,JPS98}). An example is shown in Figure \ref{fig: 23 arctic}. We consider a hexagonal region with with a staircase border in the bottom and top part and flat in the middle. Each stair is of horizontal size $s$ and vertical $m$, such that the height function becomes linear over the length of the side.\\

When $m = 1$, $s > 1$ there is no long-range property, since $1  \times 1$ squares can fit everywhere, so the shape of the region does not force any specific placement of tiles, see Figure \ref{fig: 12 arctic}.\\

\begin{figure}
 \hspace*{-1cm}
\includegraphics[scale=0.55]{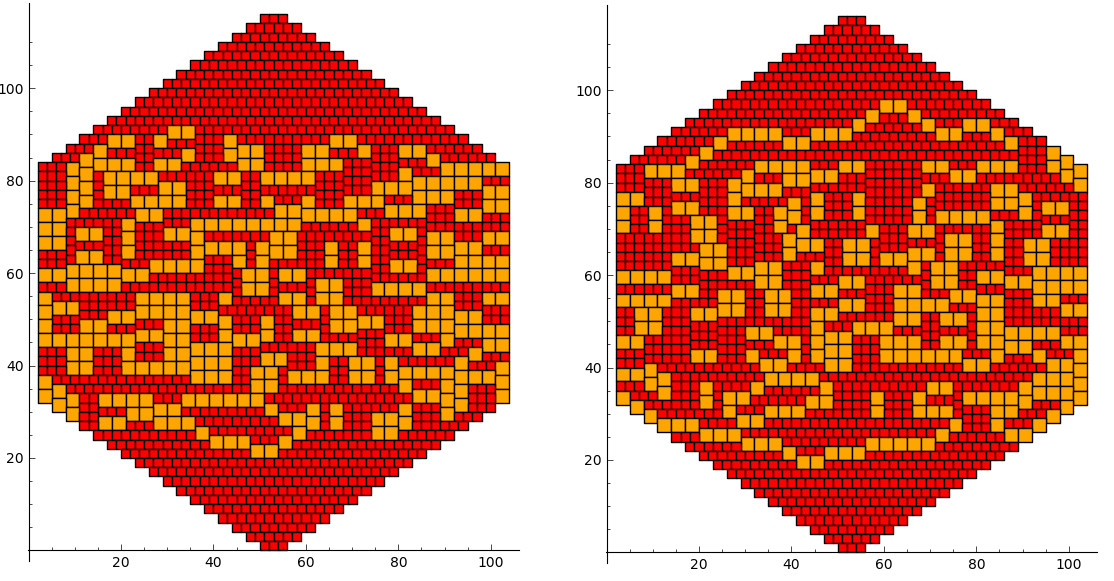}
\caption{$(2,3)-$tiling of the 50-diamond after 10 million flips (left) and 50 million flips (right).}
\label{fig: 23 arctic}
\end{figure}

\begin{figure}
\centering
\includegraphics[scale=0.55]{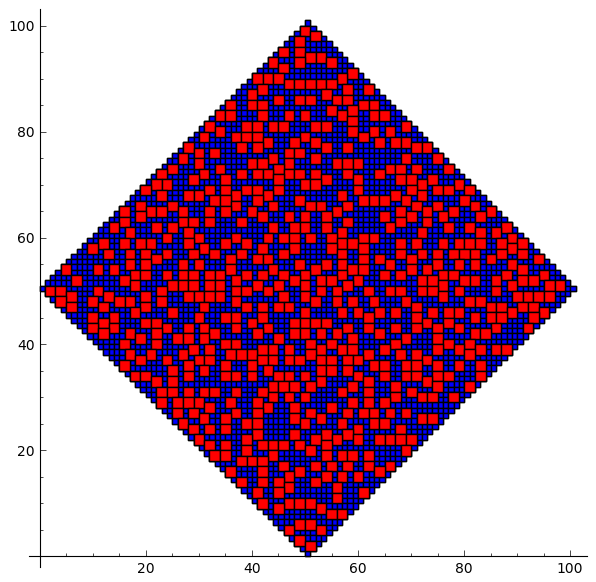}
\caption{$(1,2)-$tiling of a diamond-shaped region.}
\label{fig: 12 arctic}
\end{figure}

\section*{Acknowledgements}

The author would to thank Thomas Fernique for valuable remarks and for reading the final draft of the paper, Pavel Kalouguine, Benoît Laslier and Eric Rémila for their helpful comments.



\begin{thebibliography}{99}


\bibitem{A81} D.  Aldous, \textit{Random walks on finite groups and rapidly mixing Markov chains}, Seminaire de Probabilites 1981/1982, Springer Lecture Notes in Mathematics \textbf{986}, pp. 243--297.

\bibitem{BD97} R. Bubley, M. Dyer, \textit{Path coupling: a technique for proving rapid mixing in Markov chains}, 38th Annual Symposium on Foundations of Computer Science, 1997, pp. 223--231.
\bibitem{CLP02} H. Cohn, M. Larsen, J. Propp, \textit{The shape of a typical boxed plane partition}, New York Journal of Mathematics \textbf{4}, 1998, pp. 137--165.

\bibitem{CL90} J. Conway, J. Lagarias, \textit{Tilings with polyominos and combinatorial group theory}, Journal of Combinatorial Theory A \textbf{53}, 1990, pp. 183--208.

\bibitem{DS91} P. Diaconis, D. Stroock, \textit{Geometric bounds for eigenvalues of Markov chains}, Annals of Applied Probability \textbf{1}, 1991, pp. 36--61.
\bibitem{DG98} M. Dyer, C. Greenhill, \textit{A More Rapidly Mixing Markov Chain for Graph Colorings}, Random Structures and Algorithms \textbf{13}, 1998, pp. 285--317.
\bibitem{DG00} M. Dyer, C. Greenhill, \textit{On Markov chains for independent sets},  Journal of Algorithms \textbf{35}( Issue \textbf{1}), 2000, pp. 17--49.
\bibitem{DFJ02} M. Dyer, A. Frieze, M. Jerrum, \textit{On counting independent sets in sparse graphs}, SIAM J. Computing \textbf{33}, 2002, pp. 1527--1541.

\bibitem{JS96} M. Jerrum and A. Sinclair, \textit{The Markov chain Monte Carlo method: an approach to approximate counting and integration}, in D. Hochbaum, ed., Approximation Algorithms or NP-Hard Problems, PWS Publishing, Boston, 1996, pp. 482--520.
\bibitem{JPS98} W. Jockush, J. Propp, P. Shor, \textit{Random Domino Tilings and the Arctic Circle Theorem},  arXiv:math$/$9801068v1 [math.CO], 1998.


\bibitem{KK92} C. Kenyon, R. Kenyon, \textit{Tiling a Polygon with Rectangles}, Proc. 33rd FOCS, 1992, pp. 610--619.

\bibitem{K94} R. Kenyon, \textit{A Note on Tiling with Integer-sided Rectangles},  Journal of Combinatorial Theory \textbf{74} (Issue 2), 1996, pp. 321--332.
\bibitem{LPW09} D.A. Levin, Y. Peres, E.L. Wilmer, \textit{Markov Chains and Mixing Times}, American Mathematical Society, 2009.
\bibitem{LRS95} M. Luby, D. Randall, A. Sinclair, \textit{Markov Chain Algorithms for Planar Lattice Structures}, Proceedings in the 36th IEEE Symposium on Foundations  of Computer Science, 1995, pp. 150--159.


\bibitem{LV99} M. Luby, E. Vigoda, \textit{Fast Convergence of the Glauber Dynamics for Sampling Independent Sets: Part I}, Random Structures \& Algorithms, Volume \textbf{15}, Issue \textbf{3-4}, 1999, pp. 229--241.
\bibitem{R06} D. Randall, \textit{Rapidly Mixing Markov Chains with Applications in Computer Science and Physics}, IEEE CS and the AIP Computing in Science \& Engineering, 2006, pp. 30-41.

\bibitem{BR02} E. R\'{e}mila, \textit{Tiling a polygon with two kinds of rectangles}, Algorithms – ESA 2004, V. 3221 of the series Lecture Notes in Computer Science, pp. 568--579.
\bibitem{S92} A. Sinclair, \textit{Improved bounds for mixing rates of Markov chains and multi-commodity flow}, Combinatorics, Probability \& Computing \textbf{1}, 1992, pp. 351--370.
\bibitem{T90} W.P. Thurston, \textit{Conway's tiling groups}, The American Mathematical Monthly\textbf{ 97 (8)}, 1990, pp. 757--773.

\bibitem{W04} D.B. Wilson, \textit{Mixing Times of Lozenge Tiling and Card Shuffling Markov Chains}, The Annals of Applied Probability \textbf{14},  2004, pp. 274--325.




\end{thebibliography}
\end{document}